\documentclass[aps,prx,onecolumn,superscriptaddress,nofootinbib,11pt,floatfix]{revtex4-2}

\usepackage{mathtools}
\usepackage{bbm}
\usepackage{ulem}   
\usepackage[dvipsnames]{xcolor}
\usepackage{braket}
\usepackage{bbold}
\usepackage{comment}
\usepackage{amsthm} 
\usepackage{amssymb}    
\usepackage{bm}
\usepackage[colorlinks, linkcolor=blue, citecolor=blue]{hyperref}

\numberwithin{equation}{section}

\def\U{\mathrm{U}(1)}
\def\T{\mathcal{T}}

\def\A{\mathcal{A}}

\def\B{\mathcal{B}}

\def\E{\mathcal{E}}
\def\Z{\mathcal{Z}}
\def\cM{\mathcal{M}}

\def\bZ{\mathbb{Z}}

\makeatletter 
\def\l@subsubsection#1#2{}
\makeatother

\newcommand{\Tr}{\operatorname{Tr}}

\newtheorem{theorem}{Theorem}[section]

\theoremstyle{definition}

\begin{document}

\title{Note on proliferation transitions of non-Abelian anyons}
\author{Meng Cheng}
\affiliation{Department of Physics, Yale University, New Haven, Connecticut 06511, USA}
\begin{abstract}
    We study phase transitions out of a (2+1)d topological phase driven by the proliferation of anyons in a braided fusion subcategory. We describe a general theoretical framework for such transitions, in which the topological quantum field theory (TQFT) is coupled to dynamical matter associated with anyons in the subcategory. We interpret this construction using symmetry topological field theory (SymTFT), where the dynamics is localized on the symmetry boundary of the (3+1)d slab. We analyze several examples, and present explicit Chern-Simons-Higgs (CSH) field theories realizing the proliferation transitions. We also examine proliferation transitions that implement gauging of non-invertible one-form symmetry in the parent TQFT. This includes a general construction for ${\rm Rep}(G)$ one-form symmetry for a finite group $G$, and a CSH realization of the non-invertible one-form gauging ${\rm SU}(2)_{10}\rightarrow {\rm Spin}(5)_1$. We discuss which anyons proliferate at these transitions. 
\end{abstract}

\maketitle

\tableofcontents

\section{Introduction}

A broad class of phase transitions out of a topological phase in (2+1)d is triggered by proliferation of anyons. That is, certain anyon excitations become light to drive the transition and govern the low-energy dynamics. At low energy, the topological phase is described by a topological quantum field theory (TQFT). To describe phase transitions, it is necessary to move beyond the topological limit. A useful approach is to construct a continuum field theory, which realizes the relevant phases in different regimes in the parameter space. Following \cite{Cheng:2026qax}, we refer to such a theory as a transition field theory.

 A natural theoretical framework for transition field theories is the Chern-Simons-Higgs (CSH) theory, where the TQFT $\T$ is represented as a Chern-Simons theory~\cite{Witten:1988hf}, and the dynamical anyon is represented by a scalar field coupled to the gauge field. The scalar field transforms in a certain representation of the gauge group, such that when the field is massive it corresponds to a proliferating anyon in $\T$. Then the condensation of the scalar field Higgses the gauge field and leads to a new phase. CSH theories have been utilized to describe phase transitions of interest in condensed matter physics, see e.g. \cite{kivelson1992global, WenWu1993, chen1993mott, Barkeshli:2012rja, ClarkeNayak:2015, LeeQED3, Goldman:2019wvz, ZouQCD3, MaQCD3, song2024phase, Shi:2024pem, Ji:2026yfj, Dumitrescu:2026vre, Lu:2026fid}. We will review such CSH transitions below. 

When a single Abelian anyon proliferates, a general transition field theory has been developed in \cite{Cheng:2026qax}, where the dynamics of the anyons is described by a single complex scalar coupled to a $\U$ gauge field. The theory is based on a presentation of the starting TQFT, where the Abelian anyons under focus are represented by Wilson lines of a $\U$ gauge field. More specifically, let $a$ denote the Abelian anyon of order $n$, with one-form anomaly $p\pmod {2n}$. The presentation can be summarized as the following identity \cite{Cordova:2017vab}: 
\begin{equation}
    \T=\frac{\T \boxtimes \T_{n,p}}{(a,p,-1)}.
    \label{eq:identity_abelian}
\end{equation}
Here $\T_{n,p}$ is a $\bZ_n$ gauge theory with a Dijkgraaf-Witten term~\cite{Dijkgraaf:1989pz} labeled by $pn$, and the quotient means gauging the (anomaly-free) one-form symmetry generated by the boson $(a,p,-1)$. The transition field theory based on \eqref{eq:identity_abelian} then describes the transition from $\T$ to $\T'$ given by
\begin{equation}
    \T'=\frac{\T\boxtimes \U_{-np}}{(a,p)}.
    \label{eq:Tpr-abelian}
\end{equation}
This theory is further generalized in \cite{peelingoff} to allow for the proliferation of multiple subgroups of Abelian anyons. 

In this note, we will describe a further generalization of Eq. \eqref{eq:identity_abelian}, which singles out a \emph{braided fusion subcategory} of $\T$. This subcategory may contain both Abelian and non-Abelian anyons. The physical motivation for considering a subcategory is the following intuitive picture: in a transition field theory, the matter fields become gapless near the transition, and it is natural to postulate that the full algebra of operators generated from these fields also develops gapless correlations. Because the matter fields create anyon excitations,  this picture then suggests that an entire subcategory of anyons should proliferate.

We provide a physical explanation of the generalized identity using Symmetry Topological Field Theory (SymTFT) \cite{JiWen2020, GaiottoKulp2021,LichtmanThorngrenLindnerSternBerg2021,ApruzziBonettiGarciaEtxebarriaHosseiniSchaferNameki2023,FreedMooreTeleman2022,ChatterjeeWen2023,MoradiMoosavianTiwari2022, Bhardwaj:2023ayw, Bhardwaj:2024qiv, Bhardwaj:2025piv, Wen:2025thg}.  We then apply this construction to study phase transitions where the proliferating anyons belong to the subcategory, generalizing Eq. \eqref{eq:Tpr-abelian}. In particular, when $\T$ contains a subcategory isomorphic to that of integer-isospin anyons of ${\rm SU}(2)_k$, we provide concrete CSH theories to realize various proliferation transitions triggered by this subcategory.

We finally examine proliferation transitions connecting two TQFTs related by gauging non-invertible one-form symmetries. Such gaugings are often described as ``anyon condensation"~\cite{Bais:2008ni, Kong:2013aya}, with the intuitive picture that the ``condensation" of certain bosonic anyons causes the transformation. From the spacetime perspective, gauging means proliferating the network of topological lines labeled by a ``condensable algebra", which is a direct sum of bosonic anyons.
This picture might suggest that the anyons appearing in the condensable algebra are precisely those that proliferate at the transition. However, these objects need not themselves form a fusion subcategory, whereas we expect the light anyons in a transition field theory to be closed under fusion. We investigate this distinction in several examples,  including gauging of ${\rm Rep}(G)$ and an explicit CSH theory realizing the non-invertible gauging ${\rm SU}(2)_{10}\rightarrow {\rm Spin}(5)_1$. 

\section{Peeling off non-Abelian anyons}
\label{sec:symtft}

First we describe the non-Abelian generalization of Eq. \eqref{eq:identity_abelian}. In the following $\T$ refers to the modular tensor category (MTC) underlying the TQFT.

Suppose $\T$ has a braided fusion subcategory $\B$. Lines in $\B$ generate (possibly non-invertible) one-form symmetry.  We propose that:  
\begin{equation}
    \T=\frac{\T\boxtimes \Z(\B)}{\A_{\B}},
    \label{symtft}
\end{equation}
Let us unpack the expression. $\Z(\B)$ is the Drinfeld center of $\B$. Since $\B$ is braided, $\Z(\B)$ contains both $\B$ and $\overline{\B}$ as subcategories. Using $\B\subseteq \T$ and $\overline{\B}\subseteq\Z(\B)$, one can form a  canonical condensable algebra $\A_{\B}$ in $\B\boxtimes \overline{\B}\subseteq \T\boxtimes \Z(\B)$ given by
\begin{equation}
    \A_{\B}= \bigoplus_{x\in \B} (x, x^*).
    \label{eq:AB}
\end{equation}
Here $x$ runs over all simple objects (i.e. anyon types) of $\B$, and $x^*$ is the dual of $x$ in $\bar{\B}$.
The quotient means gauging the algebra $\A_{\B}$ \cite{Kong:2013aya, Yu:2021zmu, Kaidi:2021gbs}. 

While we do not have a completely general mathematical proof of Eq. \eqref{symtft} yet, it can be established in many nontrivial cases:
\begin{enumerate}
    \item When $\B$ is modular, $\Z(\B)=\B\boxtimes \bar{\B}$. The identity is easy to prove using the factorization property: because $\B$ is modular, we have $\T=\T'\boxtimes \B$~\cite{mueger2002structure}, where $\T'$ is another MTC. Hence $\T\boxtimes \Z(\B)=\T'\boxtimes \B\boxtimes \B\boxtimes \bar{\B}$. Because $\A_\B$ is Lagrangian in $\B\boxtimes \bar{\B}$ (i.e. the second and fourth factors), gauging it removes these two factors and results in a trivial theory. We are left with $\T'\boxtimes \B=\T$.

    \item Suppose $\B$ is a symmetric fusion category, i.e. all anyons in $\B$ braid trivially with each other~\footnote{Here, trivial braiding between anyons $a, b\in \B$ means that the S matrix element $S_{ab}=\frac{d_ad_b}{{\cal D}_\B}$, where ${\cal D}_{\B}$ is the total quantum dimension of $\B$.}. It is well-known that such a $\B$ is (super) Tannakian: $\B\simeq {\rm Rep}(G, z)$ for some finite group $G$ (see below for a more detailed description of ${\rm Rep}(G,z)$). In this case, the identity is mathematically proven in \cite{Lan:2016rcq} (see Lemma 4.18)~\footnote{We thank Yunqin Zheng for bringing Lemma 4.18 of \cite{Lan:2016rcq} to our attention.}.  
    \item When $\B$ is Abelian, it is straightforward to verify the identity by explicitly carrying out the quotient \cite{Cheng:2026qax, peelingoff}.
\end{enumerate}

Physically, it is very instructive to interpret Eq. \eqref{symtft} as an explicit construction of the SymTFT presentation of $\T$, where the one-form symmetry $\B$ is ``isolated" to the topological boundary.

To this end, denote by $\E_{\B}$ the M\"uger center of $\B$. That is, $\E_{\B}$ consists of anyons in $\B$ that braid trivially with all anyons in $\B$.  Recall that a braided fusion category with all braiding trivial is said to be ``symmetric". A mathematical theorem due to Deligne~\cite{Deligne:2002} shows that a symmetric fusion category is always equivalent to ${\rm Rep}(G,z)$, where $G$ is a finite group, and ${\rm Rep}(G,z)$ is the category of finite-dimensional representations of $G$. $z\in G$ is an order-2 central element that determines which of the objects have fermionic statistics \footnote{For more details on ${\rm Rep}(G, z)$, see Appendix G of \cite{Ellison:2024svg} for an elementary introduction.}.

To understand the SymTFT interpretation,  first note that the Drinfeld center $\Z(\B)$ is equivalent to a ``slab" of the 4d Crane-Yetter TQFT with $\B$ as the input~\cite{CraneYetter1993, CraneKauffmanYetter1997, WalkerWang2012}.  The topological line operators in the bulk are precisely described by $\E_\B={\rm Rep}(G,z)$. They generate the two-form symmetry of the 4d TQFT. When viewed as a gapped phase,  ${\rm Rep}(G, z)$ is the category of particle-like excitations. As a result, the bulk can actually be viewed as a $G$ gauge theory coupled to (bosonic or fermionic) matter forming a $G$ symmetry-protected topological phase~\cite{WangChen:2017, KongWen:2014, LanKongWen:2018, LanWen:2019, JohnsonFreyd:2022}. The nature of the SPT matter is determined by $\B$ and we will not need to know the precise relation. In addition, there are also topological surface operators generating the one-form symmetry. They can create string-like excitations.

 The top and bottom boundaries of the slab are endowed with topological boundary conditions where the bulk surface operators can terminate topologically~\footnote{This is often described as the condensation of string-like excitations.}. This condition does not uniquely fix the boundary, and we require further that line operators on the boundary correspond to $\B$ and $\bar{\B}$, respectively, on the top and bottom boundaries. Using SymTFT parlance, we treat the top boundary as the ``topological" or ``symmetry" boundary, and the bottom one as the ``physical" boundary.

Given any 3d QFT $\cal Q$ with $\B$ one-form symmetry, imagine stacking $\cal Q$ on the physical boundary, and gauge $\A_\B$. The SymTFT construction is the statement that this slab is completely equivalent to ${\cal Q}$, which amounts to the following identity:
\begin{equation}
    \frac{{\cal Q}\boxtimes \Z(\B)}{\A_{\B}}\simeq {\cal Q}.
\end{equation}
Eq. \eqref{symtft} is the special case when ${\cal Q}$ is a TQFT. A similar construction for 2d QFTs was discussed in \cite{Lin:2022dhv}.

Crucially, this presentation of $\T$ allows us to ``isolate" the $\B$ subcategory of anyons. Intuitively, since they are now on the topological boundary, one can introduce dynamical matter just on this boundary and drive phase transitions, which then describe transitions out of $\T$ when combining with the physical boundary. The only requirement is that the phase transitions on the symmetry boundary should not interfere with the bulk, so that the boundary dynamics commutes with the gauging by $\A_\B$. Equivalently, the dynamics on the boundary must preserve the one-form symmetry generated by the M\"uger center ${\rm Rep}(G,z)$. Below we will provide field theory realizations of this physical picture in various examples.

It is useful to study the (generalized) one-form symmetry of the proliferation transition. Before introducing dynamical matter on the symmetry boundary, the TQFT has one-form symmetry given by $\T$.  The topological lines on the symmetry boundary is described by the $\B$ subcategory, and those on the physical boundary is given by $\B'_\T=\frac{\T\boxtimes \overline{\B}}{\A_\B}$. Clearly $\B'_\T$ and $\B$ braid trivially with each other, and it is easy to check that ${\cal D}_{\B'_\T}\cdot {\cal D}_{\B}={\cal D}_\T$, thus $\B'_\T$ is the M\"uger centralizer of $\B$ in $\T$~\footnote{We thank Sakura Schafer-Nameki for comments on this point.}.  With the SymTFT presentation, it is clear that the topological lines on the physical boundary are not affected by the dynamics on the symmetry boundary. Thus the one-form symmetry of the proliferation transition field theory contains $\B'_\T$ as a subcategory. Then we also need to account for the one-form symmetry on the symmetry boundary not explicitly broken by the dynamical matter. As emphasized above, the latter at least includes $\E_\B$, which is identified with $\E_\B\subset \B'_\T$. In all examples discussed below in Sec. \ref{sec:examples}, no other one-form symmetry is left besides $\E_\B$, so the full one-form symmetry of the theory is given by $\B'_\T$.

\section{Review of Chern-Simons-Higgs theory}
\label{sec:CSHreview}

Suppose $G$ is a compact semisimple Lie group. A scalar field $\bm{\phi}$ transforms in a representation $R$ of $G$. The scalar may be complex or real, which we will specify whenever necessary. We consider the following Lagrangian:
\begin{equation}
    {\cal L}={\cal L}_{G_k}+|D\bm{\phi}|^2-U(\bm{\phi}),
    \label{eq:CSH}
\end{equation}
Here ${\cal L}_{G_k}$ is the Chern-Simons term of $G$ at level $k$, $D\bm{\phi}$ is the covariant derivative. The Higgs potential $U(\bm{\phi})$ is assumed to be a polynomial of $\bm{\phi}$. For simplicity, we assume there is a unique quadratic term in $U$ \footnote{This is the case for a generic complex irreducible representation. If the representation is real, then depending on whether the field $\bm{\phi}$ is real or complex, we have one or two quadratic terms. In the latter case, the unique term $|\bm{\phi}|^2$ can be selected if an additional $\U$ symmetry rotating the overall phase of $\bm{\phi}$ is imposed.}, with coefficient $r$.

Semiclassically, the theory has two phases :
\begin{enumerate}
    \item $r>0$: the scalar field $\bm{\phi}$ is massive and can be integrated out. The low-energy theory is the $G_k$ TQFT.
    \item $r<0$: the scalar field condenses and Higgses the gauge group from $G$ to a subgroup $H$, which is the stabilizer group ${\rm Stab}_{\bm{v}}(G)$ of the vacuum expectation value (VEV) $\bm{v}=\langle \bm{\phi}\rangle$. The low-energy theory is now a $H_{\tilde{k}}$ CS theory. Here the level $\tilde{k}=I_{H\hookrightarrow G}k$, where $I_{H\hookrightarrow G}$ is the Dynkin index of embedding of $H$ in $G$.
\end{enumerate}

Therefore the CSH Lagrangian \eqref{eq:CSH} realizes a transition between $G_k$ and $H_{\tilde{k}}$. The IR dynamics of the transition point, e.g. whether it flows to a conformal field theory (CFT) or a first-order transition, is beyond the scope of this work. 

It is interesting to consider the inverse mathematical question: given $G$ and $H$ where $H\subset G$ is a subgroup, is it possible to find a representation ${R}$ and a vector $\bm{v}$ in the representation space such that ${\rm Stab}_{\bm{v}}(G)=H$?

Interestingly, the existence of such $R$ and $v$ follows from the Mostow-Palais equivariant embedding theorem \cite{Mostow1957Equivariant, Palais1957Imbedding}. The theorem states that if a compact Lie group $G$ acts on a manifold $M$, then there is a $G$-equivariant embedding of $M$ into a finite-dimensional Euclidean space. In other words, $M$ can be embedded into a finite-dimensional real representation of $G$. Now apply the theorem to the homogeneous space $G/H$, which is acted on by $G$, and let $R$ be a finite-dimensional real representation of $G$, into which $G/H$ embeds. The vector $\bm{v}$ is the image of the identity coset $eH$ under the embedding map. Note that $R$ constructed this way is not necessarily irreducible, and generally not unique.  It should be emphasized that constructing a polynomial potential whose minima are precisely the $G$ orbit of $\bm{v}$ is a separate question, which does not automatically follow from the Mostow-Palais theorem. For all CSH theories studied in this work, we give explicit forms of the Higgs potentials.

Let us briefly discuss the global symmetries of the theory. In general, the zero-form symmetry depends on the representation $R$ and the potential $U$. For example, if $R$ is complex and irreducible, and $U$ only depends on $|\bm{\phi}|^2$, generically the global symmetry is the $\U$ symmetry that rotates the phase of $\bm{\phi}$. There may be additional discrete zero-form symmetries, e.g. possibly from outer automorphisms of $G$. One should therefore examine whether the Higgs phase spontaneously breaks these symmetries or not. More precisely, whether the transformation of the Higgs VEV under the global symmetry can be compensated by a gauge transformation. This is usually determined on a case-by-case basis.

The theory can also have (invertible or non-invertible) one-form symmetry~\cite{Gaiotto:2014kfa, Cordova:2025eim}. The invertible one-form symmetry is the subgroup of $Z(G)$ (the center of $G$) that acts trivially on $R$. 

To determine the full one-form symmetry, including the non-invertible ones, we make the following postulate: the Higgs field proliferates the anyon corresponding to representation $R$ in $G_k$. Consider the fusion subcategory generated by $R$ (under tensor product, direct sum and dual), which we denote by $\B_R$. A natural candidate for the one-form symmetry is the M\"uger centralizer $\B_R'$ of $\B_R$ in $G_k$. If this is indeed the case, then it makes sense to view the CSH transition as being driven by the proliferation of the subcategory $\B_R$. Recently, Ref. \cite{Cordova:2025eim} proposed a more elaborate procedure to determine the one-form symmetry of a CSH theory, and discussed examples where the one-form symmetry determined by their procedure is different from $\B_R'$. The discrepancy deserves further investigation, which we will leave for future work~\cite{Cheng_unpub}.

\section{Examples}
\label{sec:examples}

\subsection{${\rm SU}(2)_k^{\rm int}$}

To use the identity to construct transition field theories, it is convenient to have an explicit Lagrangian for $\Z(\B)$. In the following, we will focus on the example where $\B$ is the subcategory of integer isospin anyons in ${\rm SU}(2)_k$. We will denote it by $\B={\rm SU}(2)_k^{\rm int}$~\footnote{Notice that $\B$ should not be confused with ${\rm SO}(3)_k={\rm SU}(2)_k/\bZ_2$, which is only defined for $k=0\pmod 4$, or as a spin TQFT when $k=2\pmod 4$.}. Throughout this section we assume $k$ is even, otherwise $\B$ is already modular. Then $\E_{\B}=\{1, j=k/2\}$, where the $j=k/2$ line is bosonic/fermionic when $k/2$ is even/odd.

A minimal modular extension of $\B$ is just ${\rm SU}(2)_k$. Thus we have
\begin{equation}
    \Z({\rm SU}(2)_k^{\rm int})=\frac{{\rm SU}(2)_k\boxtimes \overline{{\rm SU}(2)_k}}{(\frac{k}{2}, \frac{k}{2})}.
    \label{eq:Zm}
\end{equation}
In other words, the Drinfeld center is a ${\rm SO}(4)=\frac{{\rm SU}(2)_{\rm L}\times {\rm SU}(2)_{\rm R}}{\bZ_2}$ CS theory, with level $-k,k$.  It is worth noting that there are other CS theory presentations of $\Z(\B)$, e.g. by choosing different minimal modular extensions, which we will return to later.

In the SymTFT slab, one can imagine the two ${\rm SU}(2)$ gauge fields live on the 3d boundaries, while the remaining $\bZ_2$ center propagates in the 4d bulk.  The integer isospin Wilson lines are genuine line operators on the boundary, while the half-integer isospin Wilson lines are attached to topological surface operators in the bulk. Compared with the discussions in Sec. \ref{sec:symtft}, ${\rm SU}(2)_{\rm R}$ corresponds to the symmetry boundary.

To trigger a proliferation transition, one can imagine introducing dynamical anyons carrying integer SU(2) spins on the symmetry boundary. The dynamics must commute with the quotient in Eq. \eqref{eq:Zm}. Roughly speaking, only anyons carrying integer isospin quantum numbers can become dynamical. Alternatively, the dynamics must respect the $\bZ_2$ one-form symmetry generated by the $j=k/2$ line of ${\rm SU}(2)_k$. A convenient formalism is to treat the symmetry boundary as a 3d theory with the $\E_{\B}$ one-form symmetry. 

Within the CSH framework reviewed in Sec. \ref{sec:CSHreview}, the following two scenarios for transition field theories emerge naturally:
\begin{itemize}
    \item The symmetry boundary is locally ${\rm SU}(2)_k$ coupled to scalar fields in integer isospin representations. The condensation of the field triggers a Higgs transition and the gauge group is reduced to a subgroup of ${\rm SU}(2)$. To preserve the $\bZ_2$ one-form symmetry, the subgroup always contains the $\bZ_2$ center. This is guaranteed when the scalars transform in integer isospin representations.
    \item The other scenario is that ${\rm SU}(2)_k$ is the effective low-energy theory, e.g. as the Higgs phase of a CSH theory, which has ${\rm SU}(2)$ as the unbroken gauge group.
\end{itemize}

Below we discuss examples to illustrate these two types of transition field theories.

\subsubsection*{Higgsing of $\mathrm{SU}(2)_k$ and ${\rm SO}(4)_{-k,k}$}

Consider an $\mathrm{SU}(2)_k$ CS theory, coupled to a complex scalar field $\Phi$ in the adjoint (isospin-1) representation. We assume that there is a U(1) symmetry that rotates the overall phase of $\Phi$. 

The most general potential for $\Phi$ preserving both the gauge symmetry and the global $\mathrm{U}(1)$ symmetry (up to $\Phi^4$) is
\begin{equation}
U(\Phi) = r\,\Phi^\dagger\Phi + \lambda(\Phi^\dagger\Phi)^2 + \kappa |\Phi\cdot\Phi|^2.
\label{eq:U}
\end{equation}
Notice that $|\Phi\cdot\Phi|^2\leq (\Phi^\dagger\Phi)^2$. So the potential is stable as long as $\lambda>0, \lambda+\kappa>0$. For $r>0$, the unique minimum of $U$ is $\langle\Phi\rangle=0$. In this case $\Phi$ is massive and can be integrated out. We are then left with ${\rm SU}(2)_k$. In the following we analyze what happens when $r<0$.

Let us start from the $\kappa<0$ case. The potential is minimized when $|\Phi\cdot\Phi|^2= (\Phi^\dagger\Phi)^2$. Writing $\Phi=\mathbf{u}+i\mathbf{v}$ where $\mathbf{u,v}$ are real vectors, this condition implies $\mathbf{u}\parallel \mathbf{v}$. Via a gauge transformation, we can take the minimum to be $\langle\Phi\rangle\propto(1,0,0)$. The stabilizer group is then $\U$, and the CS term becomes $\U_{2k}$. However, since the gauge-invariant charge-2 operator $\Phi\cdot\Phi$ also has a non-zero VEV, the global $\U$ is spontaneously broken and the Higgs phase is gapless due to the Goldstone mode.  Alternatively, if one uses a real isospin-1 scalar instead, there is no global $\U$ symmetry and the IR theory of the Higgs phase is just $\U_{2k}$.

For $\kappa>0$, the minimum satisfies $\Phi\cdot\Phi=0$, or $\mathbf{u}\cdot\mathbf{v}=0$ and $|\mathbf{u}|=|\mathbf{v}|$.
The condensate is therefore genuinely complex and up to gauge may be set to
$\langle\Phi\rangle \propto (1,i,0).$
The stabilizer of this condensate is the $\bZ_2$ center of $\mathrm{SU}(2)$. Thus the gauge group is broken down to $\bZ_2$, and the level-$2k$ CS term becomes a Dijkgraaf-Witten term. While $\langle \Phi\rangle$ is not invariant under the global $\U$ symmetry, the change can be compensated by applying a gauge transformation. Thus the $\U$ symmetry is not broken.  The resulting TQFT is either the $\bZ_2$ toric code for $k$ even, or the double semion for $k$ odd. 
 
The above analysis can be directly carried over to the Higgsing of the full ${\rm SO}(4)_{-k,k}$ theory, with the Higgs field in the isospin-1 representation of the right ${\rm SU}(2)_{\rm R}$.  For $\kappa<0$, the gauge group is broken down to $\frac{{\rm SU}(2)_{\rm L}\times \U}{\bZ_2}={\rm U}(2)$.
For $\kappa>0$, the unbroken gauge group is ${\rm SU}(2)_{\rm L}$. 
The topological order after Higgsing is simply
$\mathrm{SU}(2)_{-k}$.

We can now go back to the full Lagrangian and study the phase diagram. We find that when $r>0$, $\Phi$ is massive and can be integrated out. The TQFT is just $\T$.

When $r<0$, $\Phi$ condenses and Higgses ${\rm SO}(4)$. Depending on the sign of $\kappa$, there are two possibilities. If $\kappa<0$, we find 
\begin{equation}
    \T'=\frac{\T\boxtimes {\rm U}(2)_{-k,k}}{{\cal A}_{\B}},
\end{equation}
together with gapless Goldstone modes from the spontaneous symmetry breaking (SSB) of the global $\U$. Or if the Higgs field is real, the IR theory is the TQFT $\T'$.

If $\kappa>0$, we have 
\begin{equation}
    \T'=\frac{\T\boxtimes {\rm SU}(2)_{-k}}{{\cal A}_{\B}}.
\end{equation}

Let us give a concrete example. Set $k=4$, and $\T={\rm Spin}(12)_2$. It turns out that ${\rm Spin}(12)_2$ contains ${\rm SU}(2)_4^{\rm int}$ as a subcategory\footnote{Representative objects are the identity, the rank-2 symmetric tensor, and rank-4 anti-symmetric tensor.}. Our construction then leads to a transition between ${\rm Spin}(12)_2$ and $\T'=\U_4/\bZ_2={\rm Ising}^{\boxtimes 2}$.

Another way to understand this transition is to use an alternative representation of ${\rm Spin}(12)_2$ as\footnote{To see where this is from, note that gauging the anomaly-free $\bZ_2$ one-form symmetry in ${\rm Spin}(12)_2$ gives ${\rm SU}(12)_1\simeq {\rm SU}(3)_1\boxtimes \U_4$ (as TQFT). To recover we can first gauge $\bZ_2$ charge conjugation symmetry in each factor, and then take the diagonal $\bZ_2$ quotient. For ${\rm SU}(3)_1$, gauging charge conjugation leads to ${\rm SU}(2)_4$. For $\U_4$, gauging results in ${\rm Ising}^{\boxtimes 2}$. Taking the diagonal quotient then gives Eq. \eqref{eq:spin12}.}
\begin{equation}
    {\rm Spin}(12)_2\simeq \frac{{\rm SU}(2)_4\boxtimes {\rm Ising}\boxtimes {\rm Ising}}{(2, \psi,\psi)}.
    \label{eq:spin12}
\end{equation}
Hence our construction essentially isolates the ${\rm SU}(2)_4$ factor and Higgses it, leaving ${\rm Ising}^{\boxtimes 2}$.

A Lagrangian presentation of $\Z({\rm SU}(2)_k^{\rm int})$ is not unique. Below we consider a different presentation, for $k=2\pmod 4$. Recall the following fact: denote by $\cM$ a minimal modular extension of $\B$. For $k=2\pmod 4$, we can parametrize $\cM$ as follows~\cite{Bruillard2017}:
\begin{equation}
    \cM_m=\frac{{\rm SU}(2)_k \boxtimes {\rm Spin}(m)_1}{(k/2, \psi)}.
\end{equation}
Hence $\cM_m$ can be represented as a CS theory with gauge group $\frac{{\rm SU}(2)\times {\rm Spin}(m)}{\bZ_2}$ (where the $\bZ_2$ is the diagonal center). We then have
\begin{equation}
    \Z({\rm SU}(2)_k^{\rm int})=\frac{\cM_m\boxtimes \bar{\cM}_m}{(\psi, \psi)}.
\end{equation}

Next we introduce a complex scalar $\Phi^{aI}$, where $a=1,2,3$ and $I=1,2,\dots, m$, which is a bi-vector under $\mathrm{SU}(2)\times\mathrm{Spin}(m)$.  The representation is invariant under both centers.
We assume from the outset a global $\mathrm{U}(1)$ symmetry which rotates the phase of $\Phi^{aI}$.

Let us now analyze the potential of $\Phi$. First of all, the only quadratic term invariant under all (gauge and global) symmetries is $r|\Phi^{aI}|^2$. The quartic terms are chosen such that when $r<0$, the VEV takes the form
\begin{equation}
    \left\langle\Phi^{aI}\right\rangle \propto (1,i,0)^a \delta^{Im}.
    \label{eq:PhiVEV2}
\end{equation}
A detailed description of the quartic terms can be found in Appendix \ref{sec:bivector-potential}.

For the two factors, the Higgsing therefore gives
\begin{equation}
    \mathrm{SU}(2)_k\rightarrow \bZ_2 \text{ TC},
    \quad
    \mathrm{Spin}(m)_1\rightarrow\mathrm{Spin}(m-1)_1.
\end{equation}
Consequently, after Higgsing the $\mathcal M_m$ factor becomes $\mathrm{Spin}(m-1)_1$.

Since the $\overline{\mathcal M}_m$ factor is not Higgsed, we find
\begin{equation}
    \frac{{\cal M}_m\boxtimes \overline{\cal M}_{m}}{\bZ_2}\rightarrow \frac{{\rm Spin}(m-1)_1\boxtimes {\rm SU}(2)_{-k}\boxtimes {\rm Spin}(m)_{-1}}{\bZ_2\times \bZ_2}=\frac{ {\rm SU}(2)_{-k}\boxtimes {\rm Spin}(1)_{-1}}{\bZ_2}=\overline{\cal M}_{1}.
\end{equation}

Putting these factors together, we obtain~\footnote{We notice that the transition is independent of $m$. One possible scenario is that theories with different $m$ flow to the same IR fixed point. It is perhaps enlightening to consider a related example: suppose $m>1$ is an odd integer. Consider the $\big[{\rm Spin}(m)_1\boxtimes {\rm Spin}(m)_{-1}\big]/\bZ_2$ CS theory coupled to a bi-vector scalar field. When the scalar is massive, this theory describes a $\bZ_2$ toric code with $\bZ_2$ self-duality symmetry~\cite{Ji:2026yfj}. Now imagine condensing the scalar,  with the VEV proportional to the identity. The unbroken gauge group is $\bZ_2\times {\rm SO}(m)$, with a vanishing CS level. This theory describes a $\bZ_2$ toric code, but the self-duality symmetry is spontaneously broken~\cite{Ji:2026yfj}. This transition can also be described using the Ising critical theory. Hence it is likely that the field theories with different $m$ all flow to the Ising CFT in IR, together with the $\bZ_2$ toric code TQFT. In this case, the IR duality can be derived using the non-Abelian bosonization duality proposed in \cite{Metlitski_duality_2017}. }
\begin{align}
    \T' &= \frac{\T\boxtimes \overline{\cM}_1}{\A_\B}.
\end{align}

If $\T$ is ${\rm SU}(2)_k$, we have $\T'={\rm Spin(1)}_{-1}=\overline{\rm Ising}$, so the transition is from ${\rm SU}(2)_k$ to $\overline{\rm Ising}$. With the global $\U$ symmetry, the theory gives a direct transition between the two phases.

\subsubsection*{${\rm SU}(2)_k$ as the Higgs phase}

We now consider the scenario that ${\rm SU}(2)_k$ emerges as the Higgs phase of a CSH theory. 

More concretely, suppose there is a CSH theory that connects $G_{k'}$ and ${\rm SU}(2)_{k}$, where $G$ is a compact Lie group with ${\rm SU}(2)$ as a proper subgroup. Compatibility with the SymTFT bulk requires that the CSH theory has an exact $\bZ_2$ one-form symmetry. This can be achieved, if the embedding of ${\rm SU}(2)$ into $G$ is central (i.e. the image of the ${\rm SU}(2)$ center remains central and nontrivial), and the Higgs field is neutral under this $\bZ_2$ center. Once these conditions are satisfied, we have the following CSH transition for the slab alone: 
\begin{equation}
    {\rm SO}(4)_{-k,k}=\frac{{\rm SU}(2)_{-k}\boxtimes{\rm SU}(2)_k}{\bZ_2}\longleftrightarrow \frac{{\rm SU}(2)_{-k}\boxtimes G_{k'}}{\bZ_2}.
\end{equation}

Suppose a TQFT $\T$ has a subcategory isomorphic to ${\rm SU}(2)_{k'}^{\rm int}$, then we immediately find the following transition:
\begin{equation}
    \frac{\T \boxtimes {\rm SU}(2)_{-k}\boxtimes G_{k'}}{\A_{\B}\times \bZ_2}\longleftrightarrow  \T.  
\end{equation}

For an example, consider $G={\rm SU}(2)\times {\rm SU}(2)$, with levels $k_1, k-k_1$. A real bifundamental Higgs field $\Phi$ can break $G$ down to the diagonal ${\rm SU}(2)$ subgroup. It is easy to see that the $\bZ_2$ center is preserved. Hence we have a transition field theory for
\begin{equation}
    \frac{\T \boxtimes {\rm SU}(2)_{-k}\boxtimes {\rm SU}(2)_{k_1}\boxtimes {\rm SU}(2)_{k-k_1}}{\A_{\B}\times \bZ_2}\longleftrightarrow  \T.  
\end{equation}
Here the $\bZ_2$ is the diagonal center of the three ${\rm SU}(2)$ CS theories.

We can write down the full Higgs potential explicitly. It is convenient to think of $\Phi$ as a real SO(4) vector\footnote{Since $\Phi$ must be invariant under the diagonal center, there is no difference between ${\rm SU}(2)\times {\rm SU}(2)$ and SO(4) when discussing $\Phi$.}. Then the potential is the standard one:
\begin{equation}
U(\Phi)=r\Phi^2 + u(\Phi^2)^2, \quad u>0.
\end{equation}
When $r<0$, the gauge symmetry is Higgsed to SO(3), which is in fact the diagonal ${\rm SU}(2)$ once the diagonal $\bZ_2$ center is accounted for.

We describe another example where ${\rm SU}(2)_{10n}$ arises as the Higgs phase of ${\rm Spin}(5)_{n}$ CSH theory in Sec. \ref{sec:gauging}.

\subsection{Rep($G$)}
\label{sec:repG}

Another family of examples arises when $\B={\rm Rep}(G)$, where $G$ is a finite group. The SymTFT construction has a ``minimal" topological boundary, i.e. the boundary has no additional lines other than ${\rm Rep}(G)$ from the 4d bulk. Viewing the bulk as a $G$ gauge theory, a natural class of boundary dynamics is to introduce scalars transforming in irreducible representations of $G$.

In this case, it is more convenient to use an alternative presentation of $\T$, as follows: one first gauges the (non-invertible) one-form symmetry Rep$(G)$, resulting in a new TQFT $\T'=\T/{\rm Rep}(G)$. As a result of the gauging, $\T'$ is equipped with the dual 0-form $G$ symmetry\footnote{$G$ may or may not act faithfully in $\T'$. When the $G$ action is not faithful, the TQFT $\T'$ is enriched by the normal subgroup that acts trivially.}. Therefore, $\T$ is equivalent to $\T'$ coupled to a dynamical $G$ gauge field, which we denote by $\T'_G$.

In this presentation, it is natural to study the transition to $\T'_H$ for a subgroup $H\subseteq G$. This can be achieved by introducing a scalar field transforming as some representation of $G$, and a suitable potential so that the stabilizer group of the minimum is $H$. Locally, the transition field theory is completely equivalent to the Higgs transition from the $G$ gauge theory to the $H$ gauge theory. Upon gauging ${\rm Rep}(G)$, it is also equivalent to the $G\rightarrow H$ spontaneous symmetry breaking transition. 

This type of transition implements the (non-invertible) gauging in $\T$. For each $H\subset G$ up to conjugacy, one can associate a condensable algebra $\A_H$ in ${\rm Rep}(G)$, such that $\T/\A_H=\T'_H$. More explicitly, we have~\cite{Lan:2016rcq}
\begin{equation}
    \A_H={\rm Fun}(G/H)=\bigoplus_{\pi\in{\rm Irr}(G)} ({\rm dim}\, \pi^H)\pi,
\end{equation}
where $\pi^H$ is the $H$-invariant subspace of $\pi$.
For example, when $H=\bZ_1$, the algebra $\A_H=\bigoplus_{\pi \in {\rm Irr}(G)}({\rm dim}\,\pi)\pi$. 

It is instructive to consider a concrete example $G=\mathbb{D}_{2n}=\bZ_n\rtimes\bZ_2$. For simplicity we assume $n\geq 3$ is odd. Let us study the following subgroups: $H=\bZ_1, \bZ_2$ and $\bZ_n$.  Below, denote by $\pi_0$ the identity representation, $\pi_-$ the sign representation, and $\pi_{q}$ the 2-dimensional irrep carrying SO(2) charge $\pm q$, for $q=1,2,\cdots, \frac{n-1}{2}$. They label bosonic charges in $\T$.

\begin{enumerate}
    \item $H=\bZ_n$. This transition is driven by the condensation of the $\pi_-$. If we represent $\T$ as $\T/\bZ_2$ coupled to a dynamical $\bZ_2$ gauge field, the transition is equivalent to the $\bZ_2$ Higgs transition. In this case, the proliferating anyons are the ${\rm Rep}(\bZ_2)$ subcategory generated by $\pi_-$.
    \item $H=\bZ_2$. After gauging the ${\rm Rep}(\mathbb{D}_{2n})$ symmetry, the transition is equivalent to the $\mathbb{D}_{2n}\rightarrow \bZ_2$ SSB transition. We can realize it by the following Ginzburg-Landau theory of a complex scalar:
        \begin{equation}
            {\cal L}_{\bZ_2}=|D\psi|^2+r|\psi|^2 + u|\psi|^4+\lambda(\psi^n+\bar{\psi}^n).
        \end{equation}
This is just the XY theory with $\bZ_n$ anisotropy, coupled to a dynamical $\mathbb{D}_{2n}$ gauge field. For $n=3$, it is well-known that the anisotropy term is relevant and the transition becomes first-order. For $n\geq 5$, the anisotropy term is irrelevant at the XY fixed point, so the transition belongs to the gauged XY universality class. Note that the anisotropy is dangerously irrelevant, as it breaks the (emergent) $\U$ symmetry and makes the $r<0$ phase gapped.

This transition implements the gauging of $\A_{\bZ_2}=\bigoplus_{q=0}^{(n-1)/2}\pi_q$. We note that $\pi_q$ with $q=0, \dots, \frac{n-1}{2}$ are not closed under fusion. Let us consider which charges proliferate near the transition. At the transition, the order parameter $\psi$, transforming as $\pi_{q=1}$, has power-law correlations. The same is true for all its powers $\psi^\ell$, for $\ell\in \mathbb{N}$. However, ${\cal O}_-=i(\psi^n-\bar{\psi}^n)$ transforming in the sign representation $\pi_-$ is also gapless, even though $\pi_-$ does not appear in the algebra $\A_{\bZ_2}$. Thus without further qualification, one concludes that the full ${\rm Rep}(\mathbb{D}_{2n})$ ``proliferate" at the transition.  

\item $H=\bZ_1$. Generically, the $\mathbb{D}_{2n}\rightarrow \bZ_1$ SSB transition happens in two steps: $\mathbb{D}_{2n}\rightarrow \bZ_2\rightarrow \bZ_1$. To describe it we add to ${\cal L}$ a real scalar $\phi$:
\begin{equation}
    {\cal L}_{\bZ_1}={\cal L}_{\bZ_2}+(D\phi)^2+ r' \phi^2+ u' \phi^4+ w|\psi|^2\phi^2+ig\phi(\psi^n-\bar{\psi}^n)+\cdots.
\end{equation}
A direct transition requires $r=r'=0$, hence in general multicritical. If the multicritical point does exist, then both $\phi$ and $\psi$ become light, and the entire ${\rm Rep}(\mathbb{D}_{2n})$ category proliferates as well.
\end{enumerate}

Let us now consider the more general case. As discussed above, gauging $\A_H$ can be realized by a proliferation transition, which upon gauging ${\rm Rep}(G)$ is dual to a symmetry-breaking transition $G\rightarrow H$. If the IR limit of the dual transition is described by a CFT, then under reasonable general assumptions about the symmetry action, in particular that $G$ acting faithfully on local operators, one can show generally that there are primary operators of the CFT transforming in all finite-dimensional irreducible representations of $G$~\cite{Harlow:2018tng}. If $G$ does not act faithfully, then the same statement applies to the faithful quotient $\tilde{G}=G/{\rm ker}(G)$. After gauging $G$, a local operator transforming in an irreducible representation becomes an operator attached to the Wilson line labeled by the representation. Its critical correlations imply that the corresponding anyon becomes light, i.e. proliferates. Therefore, the transitions constructed in this way always proliferate the ${\rm Rep}(\tilde{G})$ (sub)category, even when the objects in the condensable algebra $\A_H$ do not form a subcategory.

\section{Gauging as proliferation}
\label{sec:gauging}

As mentioned in the Introduction, gauging (invertible or non-invertible) one-form symmetry in a TQFT can often be realized as proliferation transitions. In fact, such gauging has been called ``anyon condensation" in recent literature~\cite{Bais:2008ni, Kong:2013aya, burnell_anyon_2018}, which implies the ``condensation" of bosonic anyons involved in the gauging as a proliferation transition, not just a mathematical operation on TQFTs. However, the general question of whether gauging can always be realized by a proliferation transition (using an explicit field theory) remains an open question.

We first summarize known field-theory realizations of one-form gauging: 
\begin{itemize}
    \item When the one-form symmetry is invertible (equivalently, generated by Abelian bosons), its gauging can always be implemented as a proliferation transition. An explicit field theory is described in \cite{Cheng:2026qax}.  
\item Gauging of Rep$(G)$ can also be realized as a proliferation transition, as shown in Sec. \ref{sec:repG} up to details of the potential energy for the scalar field. 
\item Given a MTC $\B=G_k$ for $G$ a simply-connected simple Lie group, gauging the Lagrangian algebra $\A_\B$ (see Eq. \eqref{eq:AB}) in the double $\Z(\B)=G_k\boxtimes G_{-k}$ theory can often be implemented using a bi-fundamental Higgs scalar of $G\times G$: the Higgs potential can be chosen so that the residual gauge group is the diagonal $G$. Since the CS level now vanishes, the IR limit is governed by pure Yang-Mills dynamics, which is expected to flow to a trivial gapped phase~\cite{Nair:2002, Teper:1998te, AthenodorouTeper:2017}. Thus the result agrees with the gauging of $\A_\B$. Such CSH theories with $G={\rm SU}(2)$ have been recently applied to study transitions in self-dual lattice gauge theories~\cite{Ji:2026yfj, Lu:2026fid}.
\end{itemize}

     We now add to the list a new example of non-invertible one-form gauging realized by a CSH theory. We then discuss its possible generalizations. 

     The specific example we will focus on is gauging in ${\rm SU}(2)_{10}$. We label SU(2) irreps by $V_j$ where $j=0, \frac12, 1, \dots$ is the SU(2) spin. This theory has a condensable algebra $\A=V_0\oplus V_3$~\cite{kirillov2002q, Bais:2008ni, Eliens:2013epa}. Here $V_3$ is a non-Abelian boson with quantum dimension $2+\sqrt{3}$. Gauging $\A$ leads to ${\rm Spin}(5)_1$. The gauging can be understood from the perspective of rational conformal field theory: ${\rm Spin}(5)_1$ can be obtained from ${\rm SU}(2)_{10}$ by chiral algebra extension~\cite{Moore:1989yh}, by adding the isospin-3 primary to the chiral algebra. In other words, there is a conformal embedding ${\rm SU}(2)_{10}\subset {\rm Spin}(5)_1$ \cite{Schellekens:1989am, SchellekensWarner:1986}.

Now we discuss a transition field theory for this gauging. Instead of starting from the ${\rm SU}(2)_{10}$ side, we approach the transition from ${\rm Spin}(5)_1$. Note that the embedding of ${\rm SU}(2)$ into ${\rm Spin}(5)$ is principal, with an embedding index 10. We introduce a real Higgs scalar ${\Phi}$ transforming in the rank-3 symmetric traceless tensor representation ${\bf 30}$ of SO(5).

One can show that i) there is a VEV of the scalar, whose stabilizer group is precisely the principal ${\rm SU}(2)$ subgroup of ${\rm Spin}(5)$. The existence of such a minimum can be understood from the branching rule ${\bf 30}\rightarrow V_0\oplus V_3\oplus V_4 \oplus V_6$: the desired VEV is the vector in the $V_0$ subspace. ii) the VEV is the stable minimum of a Higgs potential,  whose detailed construction is given in Appendix \ref{sec:cartan}. Here we give a schematic form of the potential:
\begin{align}
    U(\Phi)=r\rho+u\rho^2 +& \lambda |C(\Phi)|^2 + \eta I_5(\Phi) + \kappa \rho^3,\\
    u, \: &\lambda, \: \kappa>0.
\end{align}
$\rho=\Phi^2$ is the full contraction of the tensor $\Phi$. The quartic term $|C(\Phi)|^2$ and the quintic term $I_5(\Phi)$ are defined in Appendix \ref{sec:cartan}. When $\eta$ and $\kappa$ are sufficiently small, $|C(\Phi)|^2$ and $I_5(\Phi)$ together select the desired vacua that break ${\rm Spin}(5)$ down to ${\rm SU}(2)$.  The quintic term is necessary to break the $\bZ_2$ symmetry $\Phi\rightarrow -\Phi$.
Importantly, the potential realizes a direct transition at the semiclassical level. 

It is also enlightening to analyze the excitations in both phases, at least semiclassically.

In the ${\rm Spin}(5)_1$ phase, the $\mathbf{30}$ representation belongs to the same superselection as the vector representation, which is the neutral fermion. Hence it is quite plausible that the transition may be associated with the proliferation of fermions.

Now we go to the Higgs phase. There are two kinds of excitations: the Higgs scalar, and the gauge fluctuations. Let us first ignore the gauge field, and consider the Goldstone modes for the symmetry breaking. The vacuum manifold is ${\rm Spin}(5)/{\rm SU}(2)$, which is 7-dimensional. In the branching of ${\bf 30}$, there is precisely one 7-dimensional subspace with $V_3$, which we identify as the Goldstone modes. The Spin(5) gauge fields transform in the Spin(5) adjoint $\mathbf{10}$, which branches under the symmetry breaking as $\mathbf{10}\rightarrow V_1\oplus V_3$. The $V_1$ representation corresponds to the unbroken SU(2) gauge field, while $V_3$ combines with the Goldstone modes to become massive W bosons. 
This is also consistent with $V_3$ being in the condensable algebra, as the ${\rm Spin}(5)$ gauge fluctuations  are in the identity superselection sector in ${\rm Spin}(5)_1$.

The remaining components of ${\bf 30}$ become massive Higgs particles. For example, the $V_0$ subspace corresponds to the ``radial" mode. Semiclassically, they should all become light near the transition. For instance, it is plausible that the $V_4$ Higgs particle becomes the isospin-4 anyon of ${\rm SU}(2)_{10}$, as it can be attached to the corresponding Wilson line.  However, to actually determine the anyon type of these massive excitations, a better understanding of the IR dynamics is needed, to resolve issues such as possible dressing by $W$ bosons. We conjecture that in fact all integer-spin anyons of ${\rm SU}(2)_{10}$ become light near the transition.

It is evident that the same construction may apply to other conformal embeddings $H_{\tilde{k}}\subset G_k$. As already discussed in Sec. \ref{sec:CSHreview}, for a given embedding $H\subset G$ the Mostow-Palais theorem provides a representation, which contains a vector that stabilizes $H$. The remaining task is to construct a potential whose minima are exactly the corresponding gauge orbit, which is left for future work.

\section{Conclusions and discussion}

In this work we propose a general framework for transition field theories to describe phase transitions out of a (2+1)d topological phase $\T$, driven by dynamical anyons in a braided fusion subcategory $\B$. The construction can be interpreted in the SymTFT perspective as phase transitions on the ``symmetry boundary". We study several examples of these proliferation transitions using CSH theories, including $\B={\rm SU}(2)_k^{\rm int}$ and ${\rm Rep}(G)$. In the latter case, we discuss proliferation transitions between two TQFTs related by gauging non-invertible one-form symmetry, and show that the proliferating anyons form a braided fusion subcategory. We also construct an explicit CSH realization of the non-invertible one-form gauging ${\rm SU}(2)_{10}\rightarrow {\rm Spin}(5)_1$. 

We close by discussing some open questions.

Most crucially, the fundamental identity Eq. \eqref{symtft} underlying the construction remains a conjecture in the most general form, although it has been established in several important cases.

Our focus on a braided fusion subcategory is motivated by the expectation that proliferating anyons should be closed under fusion. It would be valuable to determine whether this picture holds in general, for example in CSH theories. In fact the basic notion of ``proliferation of anyons" remains heuristic and requires a better formulation to address this question.

It is also important to understand whether every one-form gauging can be realized as a proliferation transition. A concrete step in this direction would be to generalize the CSH construction in Sec. \ref{sec:gauging} to other examples of conformal embeddings, such as ${\rm SU}(2)_{28}\subset ({\rm G}_2)_1$. This appears to be a well-defined model-building problem. 

While we have mainly considered continuum field theories, proliferation transitions have also been investigated in lattice models. In particular, exactly solvable models, such as string-net and Kitaev's quantum double models, can be perturbed to drive such phase transitions (e.g. see \cite{BurnellSimonSlingerland2012, Zhao:2024ilc}). Some examples have also been studied numerically~\cite{Tupitsyn:2008ah, WuDengProkofev2012, SomozaSernaNahum2021, BonatiPelissettoVicari2022, SchulzDusuelSchmidtVidal2013, SchulzDusuelMisguichSchmidtVidal2014}. Connecting transition field theories to microscopic lattice models is an important direction for future work,  see \cite{Ji:2026yfj, Dumitrescu:2026vre, Lu:2026fid} for recent examples.

\section*{Acknowledgements}

I would like to thank Nathan Seiberg for previous collaborations that inspired this work, and for extensive conversations on related topics. I'm grateful to Terry Gannon, Wenjie Ji, Dachuan Lu, Brandon Rayhaun and Sahand Seifnashri for enlightening discussions, and Yunqin Zheng and Sakura Schafer-Nameki for communications on their related work \cite{Schafer-Nameki:2026hgx}. ChatGPT 5.6 was used to assist various group theory calculations, and in particular suggested the Higgs potential in Appendix \ref{sec:cartan} and the proof strategy for Theorem \ref{theorem:cartan}. I thank the Institute for Advanced Study for hospitality. This work is partly supported by NSF Grant No. DMR-2424315. 

\appendix

\section{Higgs potentials}
In this appendix we provide additional details about Higgs potentials in the CSH examples studied in the main text.

\subsection{${\rm SO}(3)\times {\rm SO}(m)$ bi-vector Higgs scalar}
\label{sec:bivector-potential}

Here we give the potential that selects the minimum in Eq. \eqref{eq:PhiVEV2} for the ${\rm SU}(2)\times {\rm Spin}(m)$ CSH theory. Since the Higgs field is in the bi-vector representation, only the ${\rm SO}(3)\times {\rm SO}(m)$ quotient acts faithfully.

Define the quadratic tensors
\begin{align}
    N=|\Phi^{aI}|^2, \quad
    Q^{IJ}=\Phi^{aI*}\Phi^{aJ},\quad
    P^{IJ}=\Phi^{aI}\Phi^{aJ},\quad
    R^{ab}=\Phi^{aI}\Phi^{bI}.
\end{align}
All repeated indices are summed.  A convenient
$\mathrm{U}(1)\times\mathrm{SO}(3)\times\mathrm{SO}(m)$-invariant
quartic potential is
\begin{align}
    U(\Phi)
    ={}&rN+\lambda N^2
    +\alpha\left(N^2-|Q|^2\right)
    +\beta\left(N^2-|R|^2\right)
    +\gamma|P|^2.
    \label{eq:bifundamental-potential}
\end{align}
Here $|Q|^2=Q^{IJ}Q^{JI}, |P|^2=P^{IJ*}P^{IJ}, |R|^2=R^{ab*}R^{ab}$.  The last three quartic expressions are nonnegative.  We assume that $\lambda, \alpha, \beta, \gamma$ are all positive.

For $r>0$, the minimum is $\Phi=0$.

For $r<0$, the key is the $N^2-|Q|^2$ term. We will show that if $N^2=Q^{IJ}Q^{JI}$, $\Phi$ when regarded as a complex $3\times m$ matrix has complex rank one. To see this, first we notice that from the Cauchy-Schwarz inequality
\begin{equation}
    |Q^{IJ}|^2=|\bar{\Phi}^{aI}\Phi^{aJ}|^2 \leq \sum_a |\Phi^{aI}|^2\sum_b |\Phi^{bJ}|^2=N^IN^J.
\end{equation}
Here $N^I = \sum_a |\Phi^{aI}|^2$. Therefore 
\begin{equation}
    |Q|^2=\sum_{I,J}|Q^{IJ}|^2\leq \sum_{I,J}N^IN^J=\left(\sum_I N^I\right)^2=N^2.
\end{equation}
Thus the equality is reached when $|Q^{IJ}|^2=N^IN^J$, which means that $\Phi^{aI}$, as a $3$-dimensional complex vector, is parallel to $\Phi^{aJ}$. Hence the column vectors of $\Phi$ are all parallel, and the matrix has (complex) rank one.

We may therefore write $\Phi^{aI}=z^a w^I$, and the $\beta$ term becomes
\begin{equation}
    \beta|z|^4 \left(|w|^4-|w\cdot w|^2\right).
\end{equation}
It is minimized when $w$ is real up to an overall phase.  Absorbing that
phase into $z$, we can write $w=n$ with $n\in\mathbb R^m$ and
$n\cdot n=1$.  The $\gamma$ term then becomes $|z\cdot z|^2$, which imposes $z\cdot z=0$.
An $\mathrm{SO}(3)\times\mathrm{SO}(m)$ gauge
transformation therefore brings the VEV to
\begin{equation}
    \left\langle\Phi^{aI}\right\rangle
    =v_0(1,i,0)^a\delta^{Im},
    \qquad
    v_0^2=-\frac{r}{4\lambda}.
\end{equation}
Thus \eqref{eq:bifundamental-potential} selects precisely the complex
rank-one Higgs configuration assumed in the transition.

\subsection{${\rm SO}(5)$ ${\rm Sym}^3_0(\mathbf{5})$ Higgs scalar}
\label{sec:cartan}

We first explicitly describe the VEV that realizes the ${\rm SO}(5)\rightarrow {\rm SO}(3)$ symmetry-breaking pattern, where the embedding of ${\rm SO}(3)$ is principal.

Recall that the rank-2 symmetric traceless tensor representation of SO(3) is 5-dimensional: $\mathbb{R}^5={\rm Sym}^2_0(\mathbb{R}^3)$. In other words, the vector representation of SO(5) is identified with the isospin-2 representation of SO(3). This defines the principal embedding ${\rm SO}(3)\hookrightarrow {\rm SO}(5)$. The SO(3) action is $Q\rightarrow RQR^{\rm T}$ for $Q\in {\rm Sym}^2_0(\mathbb{R}^3)$. We choose a basis $\{T_a\}$ for this space, where $a=1,2,3,4,5$. They are normalized such that $\Tr T_aT_b=\delta_{ab}$. Since $\sum_a T_a^2$ is invariant under irreducible SO(3) action, it must be a scalar: $\sum_a T_a^2=\frac53\mathbf{1}$.

Define a rank-3 symmetric tensor
\begin{equation}
    d_{abc}=\Tr T_{(a}T_bT_{c)}.
\end{equation}
Here $(\cdots)$ means symmetrizing over the indices. Hence by definition $d_{abc}$ is symmetric. Now we check the trace:
\begin{equation}
 d_{aab}=\Tr\left[\left(\sum_aT_a^2\right)T_b\right]
 =\frac53\Tr T_b=0.
\end{equation}
Next we show that $d_{abc}$ is a singlet under SO(3). In fact, it is useful to define $d$ in the following way: expand $Q$ as $Q=x_aT_a$, then
\begin{equation}
    \Tr Q^3= d_{abc}x_ax_bx_c.
\end{equation}
It then immediately follows from $\Tr (RQR^{\rm T})^3=\Tr Q^3$ that $d$ is invariant under SO(3).

So far we have shown that SO(3) stabilizes $d$. Denote this principal SO(3) by $H_0$. We still need to show that the stabilizer group $H_d={\rm Stab}_d(\mathrm{SO}(5))$ is actually $H_0$. Since ${\rm SO}(3)$ is a maximal connected subgroup, $H_d$ must be a finite extension of $H_0$. As a result, $H_d$ must satisfy $H_d\subset N(H_0)$, where $N(H_0)$ is the normalizer group. For any $h\in N(H_0)$, $hH_0h^{-1}$ is an automorphism of $H_0$. Since SO(3) has no nontrivial outer automorphism, $hH_0h^{-1}$ must be an inner automorphism. It means that there exists some $h_0\in H_0$ such that $h_0h$ commutes with $H_0$. In other words, $h_0h$ is in the centralizer of $H_0$. However, $H_0$ acts irreducibly on the SO(5) vector representation, and by Schur's lemma the centralizer of $H_0$ must be a scalar, which has to be +1 as the group is SO(5). This proves that $N(H_0)=H_0$, and $H_d=H_0$ as well. Finally, the lift of the principal SO(3) in SO(5) to Spin(5) is SU(2).

Next we proceed to the construction of the Higgs potential.

Recall that the Higgs field $\Phi_{abc}$ is fully symmetric and traceless:
\begin{equation}
    \Phi_{abc}=\Phi_{(abc)}, \quad \Phi_{aab}=0.
\end{equation}
First we define the quadratic invariant 
\begin{equation}
    \rho(\Phi) = \Phi_{abc}\Phi_{abc}.
\end{equation}
To define the quartic term, first we define 
\begin{equation}
    S_{bcde}(\Phi)=\frac13\big(\Phi_{abc}\Phi_{ade} + \Phi_{abd}\Phi_{ace}+\Phi_{abe}\Phi_{acd} \big)= \Phi_{a(bc}\Phi_{de)a}.
\end{equation}
Then we remove the scalar component $S_{bcde}\delta_{(bc}\delta_{de)}=\frac23 \rho$:
\begin{equation}
    C_{bcde}(\Phi)=S_{bcde}(\Phi)-\frac{2}{35}\rho(\Phi) \delta_{(bc}\delta_{de)}.
\end{equation}
$C$ is a rank-4 symmetric tensor.

Next we show that:
\begin{theorem}
    The non-zero solutions of $C_{bcde}(\Phi)=0$ are precisely the SO(5) orbit of $\pm \sqrt{\frac{\rho(\Phi)}{\rho(d)}}d$.
    \label{theorem:cartan}
\end{theorem}

\begin{proof}
We define the following cubic homogeneous polynomial 
\begin{equation}
    P_\Phi(x)=\Phi_{abc}x_ax_bx_c, \quad x\in \mathbb{R}^5.
\end{equation}
First we prove if $C(\Phi)=0$, then $P_{\Phi}(x)$ satisfies the following equations:
\begin{equation}
    \Delta P_\Phi(x)=0,\quad |\nabla P_\Phi(x)|^2=\frac{18}{35}\rho|x|^4.
\end{equation}
Both can be verified by direct computations. We have
\begin{equation}
    \partial_aP_\Phi=3\Phi_{abc}x_bx_c.
\end{equation}
Here $\partial_a\equiv \frac{\partial}{\partial x_a}$. Then
\begin{align}
    \Delta P_\Phi(x)=\partial_a\partial_a P_\Phi=6\Phi_{aab}x_b=0.
\end{align}
and
\begin{align}
    |\nabla P_\Phi|^2&=\partial_a P_\Phi \partial_a P_\Phi\\
                     &=9\Phi_{abc}\Phi_{ade}x_bx_cx_dx_e\\
                     &=9S_{bcde}x_bx_cx_dx_e\\
                     &=9\big(C_{bcde} + \frac{2}{35}\rho \delta_{(bc}\delta_{de)}\big)x_bx_cx_dx_e\\
                     &=\frac{18}{35}\rho|x|^4.
\end{align}
In the Higgs phase, $\rho>0$, so we define 
\begin{equation}
    f(x)=\sqrt{\frac{35}{2\rho}}P_\Phi(x).
\end{equation}
It satisfies the isoparametric equations:
\begin{equation}
|\nabla f(x)|^2=9|x|^4,\quad \Delta f(x)=0.
\end{equation}
A theorem due to \'E. Cartan \cite{Cartan:1939Isoparametric} fixes the form of $f(x)$ up to an O(5) rotation:
\begin{equation}
    f(x)=x_5^3+\frac32 x_5(x_1^2+x_2^2-2x_3^2-2x_4^2)+\frac{3\sqrt{3}}{2}\big[x_4(x_2^2-x_1^2)+2x_1x_2x_3\big] = \sqrt{6}\Tr Q^3(x).
    \label{eq:feqQ3}
\end{equation}
This explicit form is given in \cite{Tkachev:2010Cartan}.  Here 
\begin{equation}
    Q(x)= \frac{1}{\sqrt{6}}
\begin{pmatrix}
-\sqrt{3}x_4-x_5 & \sqrt{3}x_3 & \sqrt{3}x_1 \\[2mm]
\sqrt{3}x_3 & \sqrt{3}x_4-x_5 & \sqrt{3}x_2 \\[2mm]
\sqrt{3}x_1 & \sqrt{3}x_2 & 2x_5
\end{pmatrix}
\end{equation}
It is obvious that $Q$ is symmetric and traceless. The normalization is chosen such that $\Tr Q^2=|x|^2$. Hence $Q(x)$ defines an isometric map from $\mathbb{R}^5$ to ${\rm Sym}_0^2(\mathbb{R}^3)$.

Now by comparing Eq. \eqref{eq:feqQ3} with the definition of $d$, we immediately see that $\Phi$ belongs to the same SO(5) orbit as $\pm d$ (up to the radial scale).
\end{proof}

We can thus use the following Higgs potential:
\begin{equation}
    U(\Phi)=r\rho+u\rho^2 + \lambda |C|^2, \quad u,\lambda>0.
\end{equation}
Here $|C|^2=C_{abcd}C_{abcd}$.  When $r<0$, the first two terms fix $\rho=-\frac{r}{2u}$. The last term is then minimized when $\Phi\propto \pm d$.

This is almost what we want, with one caveat: the Lagrangian has a $\bZ_2$ global symmetry $\Phi\rightarrow -\Phi$, which is now spontaneously broken in the Higgs phase. In order to lift this degeneracy, we need to add an O(5) pseudoscalar, i.e. a ${\rm SO}(5)$-invariant polynomial containing only terms of odd degree. A little representation theory shows that the first such term is quintic, because $\mathbf{30}^{\otimes 3}$ does not contain any singlet, while ${\rm Sym}^5(\mathbf{30})$ does. In fact there is a unique quintic term, which we denote by $I_5$. The actual expression of $I_5$ is not important to us. To stabilize the potential we also need to add a sextic term, so the full potential is 
\begin{equation}
U=r\rho + u\rho^2 + \lambda |C|^2 + \eta I_5+\kappa\rho^3, \quad u, \lambda, \kappa>0.
\end{equation}
As long as $\kappa, |\eta|$ are sufficiently small, the $\bZ_2$ symmetry is explicitly broken and the degeneracy is lifted.

\bibliography{proliferation.bib}
\end{document}